\documentclass[conference,letterpaper]{IEEEtran}

\IEEEoverridecommandlockouts

\usepackage{cite}
\usepackage{url}
\usepackage{ifthen}
\usepackage{algorithm, algorithmicx, algpseudocode}
\usepackage[pdftex]{graphicx}
\usepackage{textcomp}
\usepackage[cmex10]{amsmath}
\usepackage{amssymb,amsfonts}
\usepackage{pifont}
\newcommand{\cmark}{\ding{51}}%
\newcommand{\xmark}{\ding{55}}%
\usepackage{amsthm}
\usepackage{mathrsfs}
\def\BibTeX{{\rm B\kern-.05em{\sc i\kern-.025em b}\kern-.08em T\kern-.1667em\lower.7ex\hbox{E}\kern-.125emX}}
\usepackage{tikz}
\usetikzlibrary{automata,positioning,calc}
\usetikzlibrary{intersections}
\usetikzlibrary{decorations.pathreplacing,angles,quotes}

\usepackage{bm}
\usepackage{amscd}

\theoremstyle{definition}
\newtheorem{theorem}{Theorem}
\newtheorem{definition}{Definition}
\newtheorem{lemma}{Lemma}
\newtheorem{prop}{Proposition}
\newtheorem{cor}{Corollary}
\newtheorem{remark}{Remark}
\newtheorem{eg}{Example}

\newcommand{\one}[1]{\mbox{1}\hspace{-0.25em}\mbox{l}_{\left\{#1\right\}}}

\newcommand{\argmin}{\operatornamewithlimits{argmin}}
\newcommand{\arginf}{\operatornamewithlimits{arginf}}
\newcommand{\argsup}{\operatornamewithlimits{argsup}}
\newcommand{\relmiddle}[1]{\mathrel{}\middle#1\mathrel{}} 
\newcommand\independent{\protect\mathpalette{\protect\independenT}{\perp}}
\def\independenT#1#2{\mathrel{\rlap{$#1#2$}\mkern2mu{#1#2}}}

\def\br{\mathbb R}

\def\vE{\mathbb E}

\def\vV{\mathbb V}

\font\b=cmr10 scaled\magstep4

\def\bigzerou{\smash{\lower1.7ex\hbox{\b 0}}}
\def\bigzerou{\smash{\lower1.7ex\hbox{\b 0}}}

\begin{document}
\title{A Generalization of the Stratonovich's Value of Information and Application to Privacy-Utility Trade-off 
\thanks{This research is supported in part by Grant-in-Aid JP17K06446 for Scientific Research (C). }
}

\author{
\IEEEauthorblockN{Akira Kamatsuka}
\IEEEauthorblockA{Shonan Institute of Technology \\ 
Email: \text{kamatsuka@info.shonan-it.ac.jp}
 }
\and 
\IEEEauthorblockN{Takahiro Yoshida}
\IEEEauthorblockA{Nihon University \\ 
Email: \text{yoshida.takahiro@nihon-u.ac.jp}
  }
\and 
\IEEEauthorblockN{Toshiyasu Matsushima}
\IEEEauthorblockA{Waseda University \\ 
Email: \text{toshimat@waseda.jp}
 }
}
\maketitle

\begin{abstract}
The Stratonovich's value of information (VoI) is quantity that measure how much inferential gain is 
obtained from a perturbed sample under information leakage constraint. 
In this paper, we introduce a generalized VoI for a general loss function and general information leakage.  Then we derive an upper bound of the generalized VoI. 
Moreover, for a classical loss function, we provide a achievable condition of the upper bound which is 
weaker than that of in previous studies. Since VoI can be viewed as a formulation of a privacy-utility trade-off (PUT) problem, we provide an interpretation of the achievable condition in the PUT context.
\end{abstract}

\section{Introduction}
Research on decision-making under a constraint of information leakage has been  
studied in 1960s in the academy of sciences of the Soviet Union (USSR Academy). 
In particular, Stratonovich's work \cite{Stratonovich1965} is  pioneering, 
however, it does not appear to be widely known\footnote{Recently, his book containing this research has been translated into English \cite{belavkin2020theory}. It is worth noting that similar approach have been studied by Kanaya and Nakagawa \cite{87006}. }. 
In \cite{Stratonovich1965} and \cite{belavkin2020theory}, he introduced \textit{Value of Information} (VoI) to quantify how much inferential gain is obtained from a perturbed sample $Y$ which contains some information about original sample $X$. His formulation of the VoI was based on the Shannon's mutual information (MI) $I(X; Y)$ in the information theory \cite{shannon} and a loss (cost) function $\ell(x, a)$, where $a$ is some action (e.g. point estimation on $X$, hypothesis testing on $p_{X}$, prediction), in the statistical decision theory (see, e.g., \cite{GVK027440176}). 

Since Shannon's proposal of MI, various information leakage measures have been proposed. 
Some examples are Arimoto's MI \cite{arimoto1977}, Sibson's MI \cite{Sibson1969InformationR}, 
and Csisz\'ar's MI \cite{370121}. 
Recently, new information leakage measures have been proposed in the privacy-utility trade-off (PUT) problem, such as $f$-\textit{information} \cite{8437735} and $f$-\textit{leakage} \cite{8804205}, as privacy measures. 
In addition to these measures, by assuming a ``guessing'' adversary, information leakage measures that have operational meanings have been proposed. 
For example, Asoodeh \textit{et al.} introduced \textit{probability of correctly guessing} in \cite{8006629,8438536}. 
In \cite{8006632,7460507, 8943950}, Issa \textit{et al.} introduced \textit{maximal leakage} 
which quantifies the maximal logarithmic gain of correctly guessing any arbitrary function of the original sample. Extending the maximal leakage, Liao \textit{et al.} introduced \textit{$\alpha$-leakage} and \textit{$\alpha$-maximal leakage} in \cite{8437307,8804205,9162168,8849769}.
Liao \textit{et al.} also showed the relationships between the (maximal) $\alpha$-leakage and both Arimoto's MI and Sibson's MI. It is worth noting that Liao \textit{et al.} introduced an  $\alpha$-\textit{loss} to define the $\alpha$-leakage. 

In this study, we first introduce an information leakage measure in a general manner by extracting common properties from these specific information leakage measures. 
Then we define a generalized VoI for the information leakage measure and a 
general loss function containing the $\alpha$-loss. For the generalized VoI, 
we derive an upper bound next. 
Moreover, for a classical loss function $\ell(x, a)$, we also provide an achievable condition of the upper bound which is weaker than that of in previous studies \cite{Stratonovich1965,belavkin2020theory} and \cite{raginsky_VoI}. We also show basic properties of the achievable upper bound and some extended results.
Finally, since VoI can be viewed as a formulation of a PUT problem in a certain situation, based on our prior work \cite{9611484}, we provide an interpretation of the achievable condition in the PUT context.

\section{Preliminary}\label{sec:preliminary}
\begin{figure}[htbp]
\centering
\includegraphics[width=3.5in, clip]{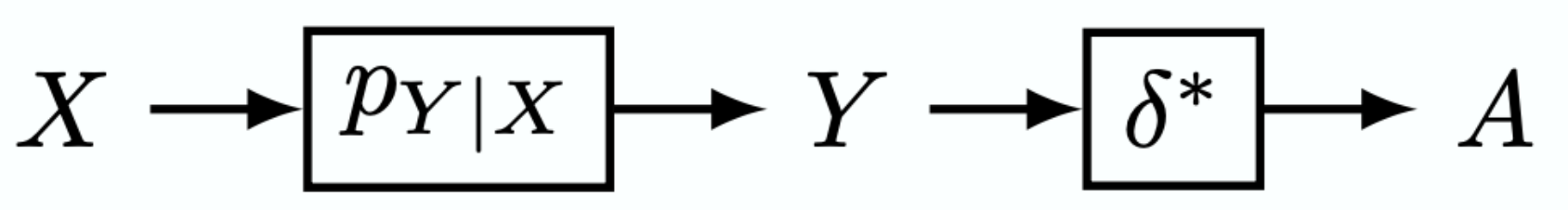}
%
%
\caption{System Model}
\label{fig:system_model}
\end{figure}

In this section, we first review the 
\text{statistical decision theory} 
and the concept of \textit{information leakage} in information theory on the system model in Figure \ref{fig:system_model}. 
For simplicity, unless otherwise stated, we will assume that all alphabets are finite.

\subsection{Notations} \label{ssec:notations}
Let $X, Y$ and $A$ be random variables on alphabets 
$\mathcal{X}, \mathcal{Y}$ and $\mathcal{A}$.
Let $p_{X, Y} = p_{X}\times p_{Y\mid X}$ and $p_{Y}$ be a given joint distribution of $(X, Y)$ 
and a marginal distribution of $Y$, respectively. 
Let $\delta^{*}\colon \mathcal{Y}\times \mathcal{A} \to [0, 1]$ 
and $\delta\colon \mathcal{Y}\to \mathcal{A}$ 
be a randomized decision rule and a deterministic decision rule, respectively. 
Since $\delta^{*}(y, a)$ is equivalent to a conditional probability $p_{A\mid Y}(a\mid y)$, 
we will use these notations interchangeably.
The classical notation for a loss function in the statistical decision theory is $\ell(x, a)$, which represents 
\textit{a loss for making an action $A=a$ when the true state is $X=x$.}
In this study, however, we extend the concept of the loss function to 
\textit{a loss for making an action $A=a$ from a sample $Y=y$ 
using the (randomized) decision rule  $\delta^{*}$ 
when the true state is $X=x$}, denoted as $\ell(x, y, a, \delta^{*})$.
Finally, we use $\log$ to represent the natural logarithm.

\subsection{Statistical decision theory}\label{ssec:SDT}
We review the basic concepts and results in the statistical decision theory next.

\begin{definition}
The loss function for a randomized decision rule $\delta^{*}\colon \mathcal{Y}\times \mathcal{A}\to [0, 1]$ is defined as 
\begin{align}
L(x, \delta^{*}(y, \cdot)) &:= \vE_{A}\left[\ell(x, y, A, \delta^{*}) \relmiddle{|} Y=y\right] \\ 
&= \sum_{a} p_{A\mid Y}(a\mid y) \ell(x, y, a, \delta^{*}).
\end{align}
\end{definition}

\begin{definition}
The risk function and the Bayes risk function for a randomized decision rule $\delta^{*}$ 
is defined as 
\begin{align}
R(x, \delta^{*}, p_{Y\mid X}) &:= \vE_{Y}[L(x, \delta^{*}(Y, \cdot))\mid X=x] \\ 
&= \sum_{y}p_{Y\mid X}(y\mid x)L(x, \delta^{*}(y, \cdot)), \\ 
r(\delta^{*}, p_{Y\mid X}) &:= \vE_{X}\left[R(X, \delta^{*}, p_{Y\mid X})\right] \\ 
&= \sum_{x}p_{X}(x)R(x, \delta^{*}, p_{Y\mid X}).
\end{align}
\end{definition}

\begin{prop}[\text{\cite[Prop 1]{9611484}}] \label{prop:opt_Bayes_randomized}
The minimal Bayes risk is given by
\begin{align}
\inf_{\delta^{*}}r(\delta^{*}, p_{Y\mid X}) &= r(\delta^{*,\text{Bayes}}, p_{Y\mid X}) \\
&= \vE_{Y}\left[\inf_{\delta(y, \cdot)} \vE_{X}\left[L(X, \delta^{*}(Y, \cdot))\relmiddle{|} Y\right]\right],
\end{align}
with the optimal randomized decision rule $\delta^{*, \text{Bayes}}$ given by 
\begin{align}
\delta^{*,\text{Bayes}}(y, \cdot) &:= \arginf_{\delta^{*}(y, \cdot)} \vE_{X}\left[L(X, \delta^{*}(y, \cdot)) \relmiddle{|} Y = y\right], \label{eq:bayes_definition}
\end{align}
where infimum is over all randomized decision rule $\delta^{*}(y, \cdot)=p_{A\mid Y}(\cdot\mid y)$ for fixed $y$.
In particular, when a channel is $p_{Y\mid X} = p_{Y}$ (i.e., $X$ and $Y$ are independent, denoted by $X\independent Y$), 
\begin{align}
\inf_{\delta^{*}}r(\delta^{*}, p_{Y}) &= r(\delta^{*,\text{Bayes}}, p_{Y}) \\ 
&= \vE_{Y}\left[\inf_{\delta^{*}(y, \cdot)}\vE_{X}\left[L(X, \delta^{*}(Y,\cdot))\right]\right].
\end{align}
\end{prop}

\begin{remark}
The corresponding result of the Proposition \ref{prop:opt_Bayes_randomized} for 
a deterministic decision rule $\delta$ and a classical loss function $\ell(x, a)$ is given by
\begin{align}
\delta^{\text{Bayes}}(y) &:= \arginf_{a} \vE_{X}\left[\ell(X, a)\mid Y=y\right], \\ 
\inf_{\delta} r(\delta, p_{Y\mid X}) &= r(\delta^{\text{Bayes}}, p_{Y\mid X}) 
= \vE_{Y}\left[\inf_{a}\vE_{X}\left[\ell(X, a)\right]\relmiddle{|} Y\right], \\ 
\inf_{\delta} r(\delta, p_{Y}) &= r(\delta^{\text{Bayes}}, p_{Y}) 
= \inf_{a}\vE_{X}\left[\ell(X, a)\right].
\end{align}
\end{remark}

\begin{table*}[h]
   \centering
   \resizebox{1.\textwidth}{!}{
  \begin{tabular}{@{} |c||c||c|c|c| @{}}
    \hline
    \multicolumn{1}{|c||}{Name} 
    & \multicolumn{1}{|c||}{Definition} 
    & \multicolumn{1}{c|}{$1)$}
    & \multicolumn{1}{c|}{$2)$}
    & \multicolumn{1}{c|}{$3)$}
    \\ \hline \hline 
    mutual information (MI) \cite{shannon}
    & $I(X; Y) := H(X) - H(X\mid Y)$  
    & \cmark 
    & \begin{tabular}{c}
    \cmark \\ 
    \cite[Thm 2.8.1]{Cover:2006:EIT:1146355}
    \end{tabular}
    & \begin{tabular}{c}
    \cmark \\ 
    \cite[Eq (2.90)]{Cover:2006:EIT:1146355}
    \end{tabular}
    \\ \hline 
    Arimoto's MI of order $\alpha$ \cite{arimoto1977}
    & $I_{\alpha}^{\text{A}}(X; Y) := H_{\alpha}(X) - H_{\alpha}^{\text{A}}(X\mid Y)$ 
    & \begin{tabular}{c}
    \cmark \\ 
    \cite[Thm 2]{arimoto1977}
    \end{tabular}
    & \begin{tabular}{c}
    \cmark \\
    \cite[Footnote 4]{8804205}, \\ 
    \cite[Cor 1]{6898022}
    \end{tabular}
    & \begin{tabular}{c}
    \cmark \\ 
    \cite[Thm 2]{arimoto1977}
    \end{tabular}
    \\ \hline 
    Arimoto's MI of order $\infty$ \cite{8804205}
    & $I_{\infty}^{\text{A}}(X; Y) :=\log \frac{\sum_{y}\max_{x}p_{X,Y}(x,y)}{\max_{x}p_{X}(x)}$ 
    & \begin{tabular}{c}
    \cmark 
    \end{tabular}
    & \begin{tabular}{c}
    \cmark \\
    \cite[Cor 1]{6898022}
    \end{tabular}
    & \begin{tabular}{c}
    \xmark \\ 
    \cite[Sec 6.6]{asoodeh_phdthesis}
    \end{tabular}
    \\ \hline 
    Sibson's MI of order $\alpha$ \cite{Sibson1969InformationR}
    & $I_{\alpha}^{\text{S}}(X; Y) := \min_{q_{Y}} D_{\alpha}(p_{X, Y} || p_{X} q_{Y})$
    & \cmark
    & \begin{tabular}{c}
    \cmark \\ 
    \cite[Eq (55)]{5707067}
    \end{tabular} 
    & \cmark
    \\ \hline 
    Sibson's MI of order $\infty$ \cite{8804205}
    & $I_{\infty}^{\text{S}}(X; Y) :=\log \sum_{y}\max_{x}p_{Y\mid X}(y\mid x)$ 
    & \begin{tabular}{c}
    \cmark 
    \end{tabular}
    & \begin{tabular}{c}
    \cmark
    \end{tabular}
    & \begin{tabular}{c}
    \cmark
    \end{tabular}
    \\ \hline 
    Csisz\'ar's MI of order $\alpha$ \cite{370121} 
    & $I_{\alpha}^{\text{C}}(X; Y) := \min_{q_{Y}} \vE_{X}\left[D_{\alpha}(p_{Y\mid X}(\cdot\mid X) || q_{Y})\right]$
    & \cmark 
    & \begin{tabular}{c}
    \cmark \\ 
    \cite[Eq (22)]{370121}
    \end{tabular}
    & \cmark 
    \\ \hline 
    $f$-information \cite{8437735}
    & $I_{f}(X; Y) := D_{f}(p_{X, Y} || p_{X} p_{Y})$
    & \cmark 
    & \begin{tabular}{c}
    \cmark \\ 
    \cite[Thm 7.9]{polyanskiy_lecnotes_fdiv}
    \end{tabular}
    & \begin{tabular}{c}
    \cmark \\ 
    \cite[Lem 4]{Csiszar:2004:ITS:1166379.1166380}, \\ 
    \cite[Thm 7.3]{polyanskiy_lecnotes_fdiv}
    \end{tabular}
    \\ \hline 
    $f$-leakage \cite{8804205} 
    & $\mathcal{L}_{f}(X\to Y) := 
    \min_{q_{Y}} D_{f}(p_{X, Y} || p_{X} q_{Y})$ 
    & \cmark 
    & \cmark 
    & \cmark
    \\ \hline \hline 
    maximal leakage \cite{8943950} 
    & $\mathcal{L}_{\text{maxL}}(X\to Y) := \sup_{U-X-Y} \log \frac{\max_{p_{\hat{U}\mid Y}}\vE_{U, Y}\left[p_{\hat{U}\mid Y}(U \mid Y) \right]}{\max_{u}p_{U}(u)}$ 
    & \begin{tabular}{c}
    \cmark \\
    \cite[Lem 1]{8943950} 
    \end{tabular}
    & \begin{tabular}{c}
    \cmark \\ 
    \cite[Lem 1]{8943950}
    \end{tabular}
    & \begin{tabular}{c}
    \cmark \\ 
    \cite[Cor 2]{8943950} 
    \end{tabular}
    \\ \hline 
    $\alpha$-leakage \cite{8804205} 
    & $\mathcal{L}_{\alpha}(X\to Y) := 
    \frac{\alpha}{\alpha-1} \log \frac{\max_{p_{\hat{X}\mid Y}} \vE_{X,Y}\left[p_{\hat{X}\mid Y}(X\mid Y)^{\frac{\alpha-1}{\alpha}}\right]}{\max_{p_{\hat{X}}}\vE_{X}\left[p_{\hat{X}}(X)^{\frac{\alpha-1}{\alpha}}\right]}$ 
    & \cmark 
    & \cmark 
    & \cmark 
    \\ \hline 
    maximal $\alpha$-leakage \cite{8804205} 
    & $\mathcal{L}_{\alpha}^{\text{max}}(X\to Y) := \sup_{U-X-Y} \mathcal{L}_{\alpha}(U\to Y)$ 
    & \cmark 
    & \begin{tabular}{c}
    \cmark \\ 
    \cite[Thm 3]{8804205}
    \end{tabular}
    & \cmark 
    \\ \hline 
    \textsf{mmse}-leakage  \text{[This study]}
    & $\mathcal{L}_{\textsf{mmse}}(X\to Y) := \vV(X) - \vE_{Y}\left[\vV(X\mid Y)\right]$ 
    & \begin{tabular}{c}
    \cmark \\ 
    \text{[Prop \ref{prop:property_mmse_leakage}]}
    \end{tabular}
    & \begin{tabular}{c}
    \cmark \\ 
    \text{[Prop \ref{prop:property_mmse_leakage}]}
    \end{tabular}
    & \begin{tabular}{c}
    \xmark \\ 
    \text{[Prop \ref{prop:property_mmse_leakage}]}
    \end{tabular}
    \\ \hline 
  \end{tabular}
  }
  \caption{Typical Information leakage measures in Information Theory}
  \label{tab:information_leakage}
\end{table*}

\subsection{Information leakage} \label{ssec:info_leakage}

In this study, we introduce information leakage measure, denoted as $\mathcal{L}(X\to Y)$,  
to quantify how much information $Y$ leak about $X$.
To this end, we extract some properties in common to well-known information leakage measures in information theory\footnote{Note that these properties are part of requirements for reasonable information leakage measures proposed by Issa \textit{et al.}\cite{8943950}.}.

\begin{definition} \label{def:info_leakage}
The information leakage $\mathcal{L}(X\to Y) = \mathcal{L}(p_{X}, p_{Y\mid X})$ is defined as a 
functional of $p_{X}$ and $p_{Y\mid X}$ that satisfies following properties:
\begin{enumerate}
\item \textit{Non-negativity}: 
\begin{align}
\mathcal{L}(X\to Y) \geq 0.
\end{align}
\item \textit{Data Processing Inequality (DPI)}: 

If $X-Y-Z$ forms a Markov chain, then 
\begin{align}
\mathcal{L}(X\to Z) \leq \mathcal{L}(X\to Y). \label{eq:DPI}
\end{align}
\item \textit{Independence}:
\begin{align}
\mathcal{L}(X\to Y) = 0 \Longleftrightarrow X \independent Y.  
\end{align}
\end{enumerate}
\end{definition}

\subsubsection{Examples of the information leakage}
Table \ref{tab:information_leakage} shows the typical information leakage measures 
in information theory that have these properties and their references, where 
\begin{itemize}
\item $\alpha \in (0, 1) \cup (1, \infty)$. Note that the value of the information leakage measures in the table are extended by continuity to $\alpha=1$ and $\alpha=\infty$.
\item $H_{\alpha}(X) := \frac{\alpha}{1-\alpha}\log \left(\sum_{x} p_{X}(x)^{\alpha} \right)^{\frac{1}{\alpha}}$ is the R\'enyi entropy of order $\alpha$.
\item $H_{\alpha}^{\text{A}}(X | Y) :=  \frac{\alpha}{1-\alpha} \log \sum_{y} \left( \sum_{x} p_{X, Y}(x, y)^{\alpha} \right)^{\frac{1}{\alpha}}$   
is Arimoto's conditional entropy of $X$ given $Y$ of order $\alpha$. 
\item $D_{\alpha}(p || q) := \frac{1}{\alpha-1} \log \left( \sum_{z} p^{\alpha}(z)q^{1-\alpha}(z) \right)$ is the R\'enyi divergence of order $\alpha$. 
\item $U$ represents an arbitrary (potentially random) function of $X$ and $\hat{U}$ represents its estimator.
\item $D_{f}(p || q) := \sum_{z\in \mathcal{\mathcal{Z}}} q(z)f \left( \frac{p(z)}{q(z)} \right)$ is the $f$-divergence, where 
$f\colon \br_{+}\to \br$ is a convex function such that $f(1) = 0$ and strictly convex at $t=1$, where $\br_{+}:=[0, \infty)$. 
\end{itemize}
Note that relationships between these information leakage measures are given as follows:
\begin{itemize}
\item $I(X; Y) = I_{1}^{\text{A}}(X; Y) = I_{1}^{\text{S}}(X; Y) = I_{1}^{\text{C}}(X; Y) 
= I_{f}(X; Y)$, where $f(t) = t\log t$. 
\item $I_{\alpha}^{\text{A}}(X; Y) = \mathcal{L}_{\alpha}(X\to Y)$ (see \cite[Thm 1]{8804205}).
\item 
$\mathcal{L}^{\text{max}}_{\alpha}(X\to Y) \notag \\ 
= 
\begin{cases}
\sup_{p_{\tilde{X}}}I_{\alpha}^{\text{A}}(\tilde{X}; Y) = \sup_{p_{\tilde{X}}}I_{\alpha}^{\text{S}}(\tilde{X}; Y), & \alpha>1, \\ 
\mathcal{L}_{\text{MaxL}}(X\to Y), & \alpha = \infty,  \\
I(X; Y), & \alpha = 1, 
\end{cases} 
$

where $p_{\tilde{X}}$ is a probability distribution over support of $p_{X}$. 
See \cite[Thm 2]{8804205} for detail.
\end{itemize}

Most of the non-negativity properties $1)$ in the Table \ref{tab:information_leakage} follow from the non-negativity of $D_{\alpha}(p || q)$ and $D_{f}(p || q)$. 
Note that properties of $\alpha$-leakage $\mathcal{L}_{\alpha}(X\to Y)$ follows from that of Arimoto's MI $I_{\alpha}^{\text{A}}(X; Y)$ because of their identity mentioned above.
Independence property $3)$ of maximal $\alpha$-leakage $\mathcal{L}_{\alpha}^{\text{max}}(X\to Y)$ follows from the property in the $\alpha$-leakage, 
while the property of Sibson's MI follows can be derived in a similar manner of \cite[Thm 2]{arimoto1977}.
Csisz\'ar's MI $I_{\alpha}^{\text{C}}(X; Y)$ and $f$-leakage $\mathcal{L}_{f}(X\to Y)$ also have the independence property $3)$.
In fact, for Csisz\'ar's MI, it follows from the non-negativity of the $\alpha$-divergence that 
$I_{\alpha}^{\text{C}}(X; Y) =\vE_{X}\left[D_{\alpha}(p_{Y\mid X}(\cdot\mid X) || q_{Y}^{\text{C, *}})\right] = 0  
\Longleftrightarrow D_{\alpha}(p_{Y\mid X}(\cdot | X) || q_{Y}^{\text{C}, *}) = 0 \textit{ a.s.}  
\Longleftrightarrow p_{Y\mid X}(y | x) = q_{Y}^{\text{C}, *}(y) = p_{Y}(y), \forall x\in \text{supp}(p_{X}), y\in \mathcal{Y}
\Longleftrightarrow X\independent Y$, 
where $q_{Y}^{\text{C}, *} := \argmin_{q_{Y}} \vE_{X}\left[D_{\alpha}(p_{Y\mid X}(\cdot\mid X) || q_{Y})\right]$, 
$\text{supp}(p_{X}):= \left\{\, x\in \mathcal{X} \relmiddle{|} p_{X}(x)>0 \right\}$ (support of $p_{X}$)
and \textit{a.s.} means almost surely.
For $f$-leakage, it can be shown in a similar way. 
Finally, DPI property of $f$-leakage follows from \cite[Lem 4]{Csiszar:2004:ITS:1166379.1166380},\cite[Thm 7.2]{polyanskiy_lecnotes_fdiv}, and a discussion in \cite[Sec V]{5707067}.

\subsubsection{\textsf{mmse}-leakage}

In addition to the typcal information leakage measures, 
we can define a new information leakage measure, 
\textit{minimum mean squared error-leakage} $\mathcal{L}_{\textsf{mmse}}(X\to Y)$, 
which has the properties $1), 2)$ but does \textit{not} satisfy $3)$ in general. 
Note that we assume that alphabets are continuous here, i.e., $\mathcal{X} = \mathcal{Y} = \br$. 

\begin{definition}[Minimum mean squared error-leakage]
The \textit{minimum mean squared error-leakage} $\mathcal{L}_{\textsf{mmse}}(X\to Y)$ is defined as 
\begin{align}
\mathcal{L}_{\textsf{mmse}}(X\to Y) 
&:= \vV(X) - \inf_{f\colon \mathcal{Y}\to \mathcal{X}} \vE_{X,Y}\left[(X - f(Y))^{2}\right] \\ 
&= \vV(X) - \vE_{Y}\left[\vV(X\mid Y)\right], 
\end{align}
where infimums is over all (measurable) function $f\colon \mathcal{Y}\to \mathcal{X}$ and 
$\vV(X):= \vE_{X}\left[(X - \vE_{X}[X])^{2}\right], \vV(X\mid Y=y):= \vE_{X}\left[(X - \vE_{X}[X\mid Y=y])^{2}\right]$ 
are variance of $X$ and conditional variance of $X$ given $Y=y$, respectively.
\end{definition}

\begin{prop} \label{prop:property_mmse_leakage}
$\mathcal{L}_{\textsf{mmse}}(X\to Y)$ has the properties $1), 2)$ but does not satisfy $3)$ in general. 
\end{prop}
\begin{proof}
Property $1)$ is trivial from the definition of the quantity. 
Property $2)$ can be proved as follows: 
If $X-Y-Z$ forms a Markov chain, then 
$\mathcal{L}_{\textsf{mmse}}(X\to Y) - \mathcal{L}_{\textsf{mmse}}(X\to Z) 
=\vE_{Y, Z}\left[ \left( \vE_{X}[X\mid Y] - \vE_{X}[X \mid Z] \right)^{2} \right] \geq 0$,
where we used the \textit{the orthogonal principle}\footnote{For any function $f(Y)$, $\vE_{X,Y}\left[(X-\vE_{X}[X\mid Y])f(Y)\right]=0$. }
in the first equality\footnote{This proof is borrowed from \cite[Thm 11]{6084749}. 
Interestingly, unlike the DPI for mutual information $I(X; Y)$, $\mathcal{L}_{\textsf{mmse}}(X\to Y) = \mathcal{L}_{\textsf{mmse}}(X\to Z)$ 
does \textit{not} imply that $Z$ is a sufficient statistic of $Y$ for $X$. 
The equality holds iff $\vE_{X}[X\mid Y] = \vE_{X}[X \mid Z] \ \textit{a.s}$. 
}.
Finally, it follows from \textit{the law of total variance}
$\vV(X) = \vE_{Y}\left[\vV(X\mid Y)\right] + \vV \left( \vE_{X}\left[X\mid Y\right]\right)$
that 
$\mathcal{L}_{\textsf{mmse}}(X\to  Y) = 0 
\Longleftrightarrow \vV \left( \vE_{Y}\left[X \mid Y\right] \right) 
= \vE_{Y}\left[ \left( \vE_{X}\left[X\mid Y\right] - \vE_{X}[X]\right)^{2}\right] = 0 
\Longleftrightarrow \vE_{X}\left[X\mid Y\right] = \vE_{X}[X] \textit{ a.s.}$ 
The equality condition $\vE_{X}\left[X\mid Y\right] = \vE_{X}[X] \textit{ a.s.}$ is often called a \textit{mean independence}, 
which is known as a weaker condition than independence $3)$, i.e., 
$X\independent Y \Longrightarrow \vE_{X}\left[X\mid Y\right] 
= \vE_{X}[X] \textit{ a.s.}$ 
\footnote{On the other hand, the mean independence is a stronger condition than \textit{uncorrelatedness}, i.e., 
$\vE_{X}[X\mid Y] \quad \textit{a.s.} \Longrightarrow \rho(X, Y) = 0$,
where $\rho(X, Y):= {(\vE_{X, Y}[XY]-\vE_{X}[X]\vE_{Y}[Y]})/{\sqrt{\vV(X)}\sqrt{\vV(Y)}]}$ is 
the coefficient of correlation between $X$ and $Y$. } 
\end{proof}

\begin{remark}
As with the \textsf{mmse}-leakage $\mathcal{L}_{\textsf{mmse}}(X\to Y)$, 
Arimoto's MI of order $\alpha=\infty$, i.e., $I_{\infty}^{\text{A}}(X; Y)$ does not have the 
independence property 3) (see \cite[Sec 6.6]{asoodeh_phdthesis}).
\end{remark}

\section{A Generalization of the Value of Information}\label{sec:gen_VoI}
In this section, we introduce \textit{the Stratonovich's Value of Information} (VoI) in a general manner 
to formulate the leakage-utility trade-off problem. 
We also show that the generalized VoI can be viewed as an analogue of \textit{the distortion-rate function} and \textit{the information bottleneck}.

\subsection{Average gain} 
We first introduce \textit{average gain} to quantify the utility of using $Y$ for a 
decision-making as largest reduction of the minimal Bayes risk compared to 
independent case.

%


\begin{definition}[Average gain] \label{def:ave_gain} 
The average gain of using $Y$ on $X$ for making an action $A$ when a loss function is $\ell(x, y, a, \delta^{*})$ 
is defined as
\begin{align}
&\textsf{gain}^{\ell}(X; Y) := \inf_{\delta^{*}} r(\delta^{*}, p_{Y}) - \inf_{\delta^{*}}r(\delta^{*}, p_{Y\mid X}) \\
&=\vE_{Y}\left[\inf_{\delta^{*}(y, \cdot)}\vE_{X}\left[L(X, \delta^{*}(Y,\cdot))\right]\right] \notag \\ 
&\quad - \vE_{Y}\left[\inf_{\delta^{*}(y, \cdot)} \vE_{X}\left[L(X, \delta^{*}(Y, \cdot))\relmiddle{|} Y\right]\right], \label{eq:ave_gain}
\end{align}
where $p_{Y}(y) := \sum_{x}p_{X}(x)p_{Y\mid X}(y\mid x)$ is a marginal distribution on $Y$. 
Note that the last equality follows from Proposition {\ref{prop:opt_Bayes_randomized}}.
In particular, the average gain with a deterministic decision rule a classic loss function $\ell(x, a)$ is given as 
\begin{align}
\textsf{gain}^{\ell}(X; Y) &= \inf_{a} \vE_{X}\left[\ell(X, a)\right] \notag \\
&\quad - \vE_{Y}\left[\inf_{a} \vE_{X}\left[\ell(X, a) \mid Y\right] \right]. 
\end{align}
\end{definition}

\begin{remark}
Note that the average gain is a statistical decision-theoretic counterpart of \textit{the average cost gain} $\Delta C$ defined in \cite{6483382}.
\end{remark}

{Using the similar argument as in \cite[Sec V.F]{9064819}, it follows that 
the average gain satisfies the DPI.
\begin{prop}[\text{\cite[Sec V.F]{9064819}}]
For any loss function $\ell(x, y, a, \delta^{*})$, 
the average gain $\textsf{gain}^{\ell}(X; Y)$ satisfies DPI.
\end{prop}}

\begin{eg}  \label{eg:sq_loss}
When a decision maker's action is to estimate $X$ deterministically under a squared-loss, i.e., 
$A=\hat{X}=\delta(Y), \ell_{\text{sq}}(x, \hat{x}) := (x-\hat{x})^{2}$, 
$\textsf{gain}^{\ell_{\text{sq}}}(X; Y) = \mathcal{L}_{\textsf{mmse}}(X\to Y)$. 
\end{eg}

\begin{eg} \label{eg:alpha_loss}
When a decision maker's action is to estimate $X$ randomly under an $\alpha$-loss proposed by Liao \textit{et al.} in \cite[Def 3]{8804205}\footnote{Technically,  Liao \textit{et al.} call $L_{\alpha}(x, \delta^{*}(y, \cdot))
:=\vE_{\hat{X}}[\ell_{\alpha}(x, y, \hat{X}, \delta^{*}) \mid Y=y]$ itself as $\alpha$-loss.
Note that the value of $L_{\alpha}(x, \delta^{*}(y, \cdot))$ is extended by continuity to $\alpha=1$ and $\alpha=\infty$.}, i.e., 
$A=\hat{X}, \ell_{\alpha}(x, y, \hat{x}, \delta^{*}) := \frac{\alpha}{\alpha-1} \left(1 - \delta^{*}(y, \hat{x})^{\frac{-1}{\alpha}}\one{\hat{x}=x} \right)$, 
\begin{align}
&\textsf{gain}^{\ell_{\alpha}}(X; Y) \notag \\
&=\begin{cases}
\frac{\alpha}{\alpha-1} \left( e^{\frac{1-\alpha}{\alpha} \cdot H_{\alpha}^{\text{A}}(X\mid Y)} 
- e^{\frac{1-\alpha}{\alpha}\cdot H_{\alpha}(X)}  \right), & \alpha>1 \\ 
H(X) - H(X\mid Y)=I(X; Y), & \alpha=1.
\end{cases}\label{eq:gain_eg} 
\end{align}
where \eqref{eq:gain_eg} follows from \cite[Lem 1]{8804205}.

\end{eg}

Intuitively, the optimal decision rule \eqref{eq:bayes_definition} seems not to depend on $y$ 
when the independent channel $p_{Y}$ is used, 
however, it is not the case in general loss function $\ell(x, y, a, \delta^{*})$. 
Thus we restrict the loss function to the following {\textit{standard loss}} class.

\begin{definition}[Standard loss]
The loss function $\ell(x, y, a, \delta^{*})$ is said to be a \textit{standard loss} if there exists a function 
$\tilde{\ell}\colon \mathcal{X}\times \mathcal{A} \times [0, 1]\to \br_{+}; (x, a, p)\mapsto \tilde{\ell}(x, a, p)$ such that 
for all $x, y, a$ and $\delta^{*}$, 
\begin{align}
\ell(x, y, a, \delta^{*}) &= \tilde{\ell}(x, a, \delta^{*}(y, a)).
\end{align}
\end{definition}

\begin{eg}
The classical loss function $\ell(x, a)$ and the $\alpha$-loss $\ell_{\alpha}(x, y, \hat{x}, \delta^{*})$ in the Example \ref{eg:alpha_loss} are
typical examples of the standard loss.
\end{eg}

\begin{prop}
For a standard loss $\ell(x, y, a, \delta^{*})$, 
the optimal decision rule \eqref{eq:bayes_definition} does not depend on $y$ when a channel is independent.
\end{prop}
\begin{proof} Since 
\begin{align}
&\inf_{\delta^{*}(y, \cdot)} \vE_{X}[L(X, \delta^{*}(y,\cdot))] \notag \\
&= \inf_{\delta^{*}(y, \cdot)}\sum_{x} p_{X}(x)\sum_{a}\delta^{*}(y, a)\tilde{\ell}(x, a, \delta^{*}(y, a))
\end{align}
is constant regardless of the value of $y$, the optimal decision rule \eqref{eq:bayes_definition} does not depend on $y$. 
\end{proof}

\subsection{A Generalization of the Value of Information}

We define VoI for information leakage to formulate the leakage-utility trade-off problem.
In the following, we assume that the information leakage $\mathcal{L}(X\to Y)$ is bounded above, i.e., 
there exists an upper bound $K(X)$ that can depend on $p_{X}$ such that for all $p_{Y\mid X}$, $\mathcal{L}(X\to Y) \leq K(X)$.

\begin{definition} \label{def:VoI_for_information_leakage} 
Let the loss function $\ell(x, y, a, \delta^{*})$ be a standard loss.
For $0\leq R\leq K(X)$, 
the \textit{generalized value of information for information leakage $\mathcal{L}(X\to Y)$} is defined as 
\begin{align}
&\textsf{V}_{\mathcal{L}}^{\ell}(R; \mathcal{Y}) 
:= \sup_{\substack{p_{Y\mid X}\colon \\ 
{\mathcal{L}(X\to Y)}\leq R}} \textsf{gain}^{\ell}(X; Y) \\ 
&= 
\inf_{\delta^{*}(y, \cdot)} \vE_{X}\left[L(X, \delta^{*}(y,\cdot))\right]
\notag \\
&\quad - \inf_{\substack{p_{Y\mid X}\colon \\ 
\mathcal{L}(X\to Y)\leq R}} \vE_{Y}\left[\inf_{\delta^{*}(y, \cdot)} \vE_{X}\left[L(X, \delta^{*}(Y, \cdot))\relmiddle{|} Y\right]\right]. 
\label{eq:VoI_leakage}
\end{align}
\end{definition}
In particular, VoI for a deterministic decision rule and a classical loss function $\ell(x, a)$ is given as
\begin{align}
\textsf{V}_{\mathcal{L}}^{\ell}(R; \mathcal{Y})&=\inf_{a} \vE_{X}\left[\ell(X, a)\right] \notag \\
&- \inf_{\substack{p_{Y\mid X}\colon \\ 
{\mathcal{L}_{\alpha}(X\to Y)}\leq R}}\vE_{Y}\left[\inf_{a} \vE_{X}\left[\ell(X, a) \mid Y\right] \right]. 
\label{eq:gen_VoI_for_classic_loss}
\end{align}

\begin{remark}
Stratonovich's original formulation of VoI is when $\mathcal{L}(X\to Y) = I(X; Y)$ and classical loss $\ell(x, a)$. 
Note that the second term of the generalized VoI  
$U(R; \mathcal{Y}) := \inf_{\substack{p_{Y\mid X}\colon \\ 
\mathcal{L}(X\to Y)\leq R}} \vE_{Y}\left[\inf_{\delta^{*}(y, \cdot)} \vE_{X}\left[L(X, \delta^{*}(Y, \cdot))\relmiddle{|} Y\right]\right]$ 
will be the \textit{distortion-rate function} $D(R; \mathcal{Y})$ under a non-standard loss function $\ell(x, y, a, \delta^{*}) = d(x, y)$, 
where $d(x, y)$ is a distortion function, which is not appropriate loss for a 
decision-making context since it only measures the distortion between $x$ and $y$.
\end{remark}

\begin{eg} From the Example \ref{eg:sq_loss} and Example \ref{eg:alpha_loss}, 
it follows immediately that 
\begin{align}
\textsf{V}_{\mathcal{L}_{\textsf{mmse}}}^{\ell_{\text{sq}}} (R; \mathcal{Y}) &= R, \qquad 0\leq R \leq \vV(X), \\ 
\textsf{V}_{I}^{\ell_{\alpha=1}}(R; \mathcal{Y}) &= R, \qquad 0\leq R \leq H(X), 
\end{align}
for all alphabet $\mathcal{Y}$.

\end{eg}

\begin{eg} 
When an action is to estimate $U$ correlated only with $X$, i.e., $A=\hat{U}$ under $\alpha=1$-loss 
$\ell_{\alpha=1}^{U}(u, y, \hat{u}, \delta^{*}) := \frac{\alpha}{\alpha-1} \left(1 - \delta^{*}(y, \hat{u})^{\frac{-1}{\alpha}}\one{\hat{u}=u} \right)$ 
and the information leakage constraint $\mathcal{L}(X\to Y) = I(X; Y)\leq R$, 
the generalized VoI is given as 
\begin{align}
\textsf{V}_{I}^{\ell_{\alpha=1}^{U}}(R; \mathcal{Y}) 
&:=  \sup_{\substack{p_{Y\mid X}\colon \\ 
{I(X; Y)}\leq R}} \textsf{gain}^{\ell_{\alpha=1}^{U}}(U; Y) \\ 
&= \sup_{\substack{p_{Y\mid X}\colon \\ 
{I(X; Y)}\leq R}} I(U; Y).
\end{align}
Note that this quantity is the well-known \textit{information bottleneck} \cite{Tishby99theinformation}.
\end{eg}

\section{Main results} \label{sec:main_result}
The main results of this paper are an upper bound of the VoI for a standard loss and 
a fundamental limit of the VoI for a classical loss.

\subsection{Upper bound and Fundamental Limit}
For a standard loss $\ell(x, y, a, \delta^{*})$, following upper bound holds.
\begin{prop} \label{prop:ub_for_information_leakage}
For a standard loss $\ell(x, y, a, \delta^{*})$, define a function as follows:
\begin{align} 
\bar{\textsf{V}}_{\mathcal{L}}^{\ell}(R; \mathcal{Y}) 
&:=\inf_{\delta^{*}(y, \cdot)} \vE_{X}\left[L(X, \delta^{*}(y,\cdot))\right] \notag \\ 
&\quad -\displaystyle \inf_{\substack{p_{Y\mid X}, \delta^{*}\colon \\ 
{\mathcal{L}(X\to A)} \leq R}} \vE_{X, Y}\left[L(X, \delta^{*}(Y, \cdot))\right].
\label{eq:gen_VoI_for_std_loss}
\end{align}
Then $\bar{\textsf{V}}_{\mathcal{L}}^{\ell}(0) = 0$ and for $0\leq R\leq K(X)$ and 
arbitrary alphabet $\mathcal{Y}$, 
\begin{align}
{\textsf{V}}_{\mathcal{L}}^{\ell}(R; \mathcal{Y}) 
\leq \bar{\textsf{V}}_{\mathcal{L}}^{\ell}(R; \mathcal{Y}). 
\label{eq:ub_information_leakage_for_std_loss}
\end{align}
\end{prop}
\begin{proof}
See Appendix \ref{proof:ub_for_information_leakage}.
\end{proof}

Note that the upper bound \eqref{eq:ub_information_leakage_for_std_loss} still 
depends on the alphabet $\mathcal{Y}$. 
Interestingly, when it comes to the classical loss function $\ell(x, a)$, 
corresponding upper bound is independent on the alphabet $\mathcal{Y}$ and it is even achievable.
\begin{theorem} \label{thm:main_result_for_information_leakage}
For a classical loss $\ell(x,a)$, define a function as follows:
\begin{align} 
\textsf{V}_{\mathcal{L}}^{\ell}(R) 
&:=
\inf_{a} \vE_{X}\left[\ell(X, a)\right] 
-  \displaystyle \inf_{\substack{p_{A\mid X}\colon \\ 
{\mathcal{L}(X\to A)} \leq R}} \vE_{X, A}\left[\ell(X, A)\right].
\label{eq:gen_VoI_fundamental_for_classic_loss}
\end{align}
Then $\textsf{V}_{\mathcal{L}}^{\ell}(0) = 0$ and 
for $0\leq R\leq K(X)$ and arbitrary alphabet $\mathcal{Y}$, 
\begin{align}
\textsf{V}_{\mathcal{L}}^{\ell}(R; \mathcal{Y}) 
\leq \textsf{V}_{\mathcal{L}}^{\ell}(R). 
\label{eq:main_result_for_information_leakage}
\end{align}
Moreover, let $t(A)$ be a \textit{sufficient statistic of $A$ for $X$} and $t(\mathcal{A})$ be a set of all values of the statistic.
Then the equality in the inequality \eqref{eq:main_result_for_information_leakage} holds 
when $\mathcal{Y}= t(\mathcal{A})$ and the optimal mechanism is given by 
\begin{align}
p^{*}_{Y\mid X}(y\mid x) := \sum_{a}p^{*}_{A\mid X}(a\mid x) \one{y = t(a)}, 
\end{align}
where $p^{*}_{A\mid X} = \arginf_{p_{A\mid X}\colon \mathcal{L}(X\to A)\leq R}\vE_{X, A}\left[\ell(X, A)\right]$. 

The statement above can be summarized as follows: 
\begin{align}
\sup_{\mathcal{Y}} \textsf{V}_{\mathcal{L}}^{\ell}(R; \mathcal{Y}) &= \textsf{V}_{\mathcal{L}}^{\ell}(R).
\end{align}
\end{theorem}

\begin{proof}
See Appendix \ref{proof:main_result_for_information_leakage}.
\end{proof}

\begin{remark}
Stratonovich call $\textsf{V}_{I}^{\ell}(R)$ 
as \textit{Value of Shannon's Information} in \cite[Chapter. 9.3]{belavkin2020theory}. 
Thus we call 
$\textsf{V}_{I_{\alpha}^{\text{A}}}^{\ell}(R)$ 
(resp. $\textsf{V}_{I_{\alpha}^{\text{S}}}^{\ell}(R), \textsf{V}_{I_{\alpha}^{\text{C}}}^{\ell}(R), \textsf{V}_{I_f}^{\ell}(R)$) 
and $\textsf{V}_{\mathcal{L}_{\alpha}}^{\ell}(R)$ (resp. $\textsf{V}_{\mathcal{L}_{\alpha}^{\text{max}}}^{\ell}(R), \textsf{V}_{\mathcal{L}_{f}}^{\ell}(R)$)
as \textit{Value of Arimoto's (resp. Sibson's, Csisz\'ar's, $f$-) Information} and \textit{Value of $\alpha$- (resp. maximal $\alpha$-, $f$-) leakage}. 
\end{remark}

Let the alphabet $\mathcal{X}$ be $\mathcal{X} := \left\{1, 2, \dots, m \right\}$ and 
$\mathcal{P}(X)$ be a probability simplex in $\br^{m}$. 
In Storatonovich's original proof of the achievability, he showed the equality condition as 
$\mathcal{Y} = \mathcal{P}(X)$ and $Y = (p_{X\mid A}(1 \mid A), p_{X\mid A}(2 \mid A), \dots, p_{X\mid A}(m \mid A)) \in \mathcal{P}(X)$. 
In \cite{raginsky_VoI}, Raginsky gave much shorter proof with $\mathcal{Y} = \mathcal{A}$ and $Y=A$. 
Note that both equality conditions are special cases of the Theorem \ref{thm:main_result_for_information_leakage}, i.e., 
following holds.
\begin{prop} \label{prop:sufficient_statistic}
$t(A) = A$ is a sufficient statistic of $A$ for $X$. 
Moreover, if a family of distributions $\{p_{A\mid X}(\cdot \mid x)\}_{x\in \mathcal{X}}$ have the same support, then 
$t(A) = (p_{X\mid A}(1 | A), p_{X\mid A}(2 | A), \dots, p_{X\mid A}(m \mid A))$ is also sufficient for $X$.
\end{prop}
\begin{proof}
See Appendix \ref{proof:sufficient_statistic}.
\end{proof}

\begin{remark}
Even though $\textsf{mmse}$-leakage $\mathcal{L}_{\textsf{mmse}}(X\to Y)$ and Arimoto's MI of order $\alpha=\infty$, i.e.,  $I_{\infty}^{\text{A}}(X; Y)$ does not have the independence property $3)$, 
almost the same result holds for $\textsf{V}_{\mathcal{L}_{\textsf{mmse}}}^{\ell}(R)$ and $\textsf{V}_{I_{\infty}^{\text{A}}}^{\ell}(R)$ 
since the only part that we use the independence property is to prove $\textsf{V}_{\mathcal{L}}^{\ell}(0) = 0$. 
Note that $\textsf{V}_{\mathcal{L}_{\textsf{mmse}}}^{\ell}(0)\geq 0$ and $\textsf{V}_{I_{\infty}^{\text{A}}}^{\ell}(0)\geq 0$ in general.
\end{remark}

\subsection{Basic properties of the Fundamental Limit}\label{ssec:basic_property_gen_VoI}
The following basic properties hold for the fundamental limit $\textsf{V}_{\mathcal{L}}^{\ell}(R)$.
\begin{prop} \label{prop:basic_property_gen_VoI}
\ 
\begin{enumerate}
\item $\textsf{V}_{\mathcal{L}}^{\ell}(R)$ is non-decreasing in $R$.
\item $\textsf{V}_{\mathcal{L}}^{\ell}(R)$ is concave (resp. quasi-concave) 
if $\mathcal{L}(X\to A)$ is convex (resp. quasi-convex) in $p_{A\mid X}$. 
\item Let $\mathcal{L}_{1}(X\to Y), \mathcal{L}_{2}(X\to Y)$ be information leakage measures. 
If there exists a constant $c>0$ such that $\mathcal{L}_{1}(X\to Y) \leq c\mathcal{L}_{2}(X\to Y)$, then 
\begin{align}
\textsf{V}_{\mathcal{L}^{(2)}}^{\ell}(R) \leq \textsf{V}_{\mathcal{L}^{(1)}}^{\ell}(cR), \\ 
\textsf{V}_{\mathcal{L}^{(2)}}^{\ell}(R/c) \leq \textsf{V}_{\mathcal{L}^{(1)}}^{\ell}(R).
\end{align}
\end{enumerate}
\end{prop}
\begin{proof}
See Appendix \ref{proof:basic_property_gen_VoI}. 
\end{proof}

\begin{cor} From the property $2)$ above, following holds.
\begin{itemize}
\item $\textsf{V}^{\ell}_{I}(R)$ is concave 
since $I(X; A)$ is convex in $p_{A\mid X}$ for fixed $p_{X}$ (see, e.g., \cite[Thm 2.7.4]{Cover:2006:EIT:1146355})
\item $\textsf{V}^{\ell}_{\mathcal{L}_{\alpha}}(R) = \textsf{V}^{\ell}_{I_{\alpha}^{\text{A}}}(R)$ is quasi-concave 
since $\mathcal{L}_{\alpha}(X\to A) = I_{\alpha}^{\text{A}}(X; A)$ is quasi-convex in $p_{A\mid X}$ for fixed $p_{X}$  (see \cite[Footnote 3]{8804205})
\item For $\alpha > 0$, $\textsf{V}^{\ell}_{I_{\alpha}^{\text{S}}}(R)$ is quasi-concave 
since $I_{\alpha}^{\text{S}}(X; A)$ is quasi-convex in $p_{A\mid X}$ for fixed $p_{X}$. 
For $0 < \alpha\leq 1$, $\textsf{V}^{\ell}_{I_{\alpha}^{\text{S}}}(R)$ is concave 
since $I_{\alpha}^{\text{S}}(X; A)$ is convex in $p_{A\mid X}$ for fixed $p_{X}$ (see \cite[Thm 10]{7282554})
\item For $0 < \alpha\leq 1$, $\textsf{V}^{\ell}_{I_{\alpha}^{\text{C}}}(R)$ is concave 
since $I_{\alpha}^{\text{C}}(X; A)$ is convex in $p_{A\mid X}$ for fixed $p_{X}$ (see \cite[Thm 9 (c)]{e23020199})
\item $\textsf{V}^{\ell}_{I_{f}}(R)$ and $\textsf{V}^{\ell}_{\mathcal{L}_{f}}(R)$ are concave 
since $I_{f}(X; A)$ and $\mathcal{L}_{f}(X\to A)$ are both convex in $p_{A\mid X}$\footnote{From the convexity of $f$-divergence \cite[Lem 4.1]{Csiszar:2004:ITS:1166379.1166380}, 
one can derive the convexity of $I_{f}(X; A)$ and $\mathcal{L}_{f}(X\to A)$ in $p_{A\mid X}$. } 
for fixed $p_{X}$ 
\item For $\alpha > 0$, $\textsf{V}^{\ell}_{\mathcal{L}_{\alpha}^{\text{max}}}(R)$ is quasi-concave 
since $\mathcal{L}_{\alpha}^{\text{max}}(X\to A)$ is quasi-convex in $p_{A\mid X}$ for fixed support of $p_{X}$ (see \cite[Thm 3]{8804205}). 
For $0< \alpha \leq 1$, $\textsf{V}^{\ell}_{\mathcal{L}_{\alpha}^{\text{max}}}(R)$ is concave 
since $\mathcal{L}_{\alpha}^{\text{max}}(X\to A)$ is convex in $p_{A\mid X}$ for fixed support of $p_{X}$
\footnote{Convexity of $\mathcal{L}_{\alpha}^{\text{max}}(X\to A)$ in $p_{A\mid X}$ follows from \cite[Thm 2]{8804205} and \cite[Thm 10]{7282554}.}
\end{itemize}
\end{cor}

Figure \ref{fig:VoI_alpha} shows a graph of the value of Shannon's information.
\begin{figure}[htbp]
\centering
\includegraphics[width=1.5  in,keepaspectratio,clip]{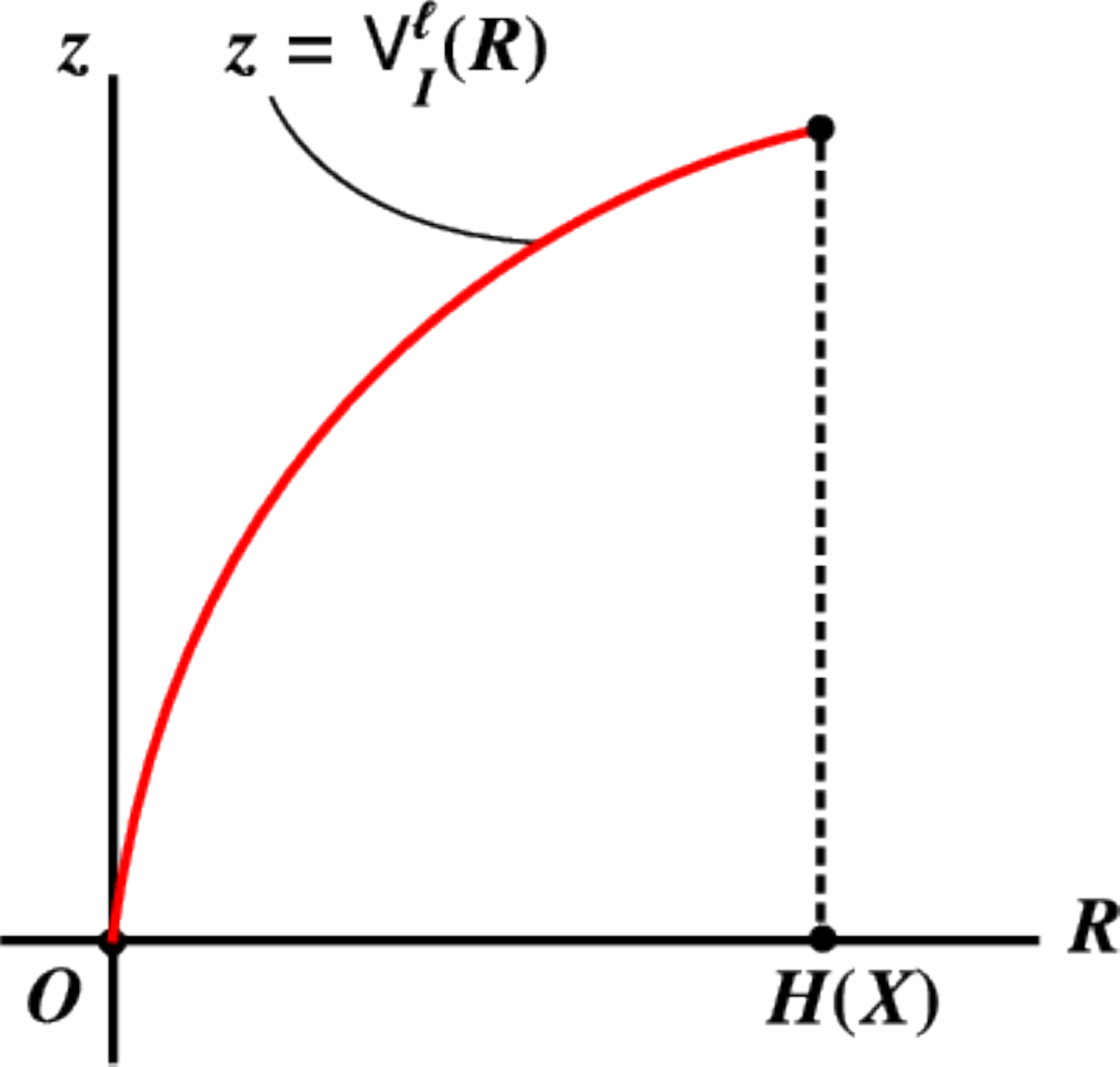}
\caption{Value of Shannon's information}
\label{fig:VoI_alpha}
\end{figure}

\subsection{Extension: logarithmic value of information}
Instead of the average gain in Definition \ref{def:ave_gain}, 
we can consider \textit{logarithmic gain} to capture utility.

\begin{definition} 
The \textit{logarithmic gain} of using $Y$ on $X$ for making an action $A$ when a loss function is $\ell(x, a)$ 
and the \textit{logarithmic value of information} are defined as follows:
\begin{align}
\textsf{Lgain}^{\ell}(X; Y) 
&:= \log \frac{\inf_{\delta^{*}}r(\delta^{*}, p_{Y})}{\inf_{\delta^{*}}r(\delta^{*}, p_{Y\mid X})} \\ 
&= \log \frac{\inf_{a}\vE_{X}\left[\ell(X, a)\right]}{\vE_{Y}\left[\inf_{a}\vE_{X}\left[\ell(X, a)\mid Y\right]\right]}, \\ 
{\textsf{LV}}_{\mathcal{L}}^{\ell}(R; \mathcal{Y}) 
&:= \sup_{\substack{p_{Y\mid X}\colon \\ \mathcal{L}(X\to Y)\leq R}} 
\textsf{Lgain}^{\ell}(X; Y). 
\end{align}
\end{definition}
\begin{eg} \label{eg:sq_loss_log} 
Let $A = \hat{X}$ and $\ell_{\text{sq}}(x, \hat{x}) = (x-\hat{x})^{2}$. Then 
\begin{align}
\textsf{Lgain}^{\ell_{\text{sq}}}(X; Y) &= \log \frac{\vV(X)}{\vE_{Y}\left[\vV(X\mid Y)\right]} =: \mathcal{L}_{\textsf{MS}}(X\to Y).
\end{align}
From Proposition \ref{prop:property_mmse_leakage}, it follows that 
$\textsf{Lgain}^{\ell_{\text{sq}}}(X; Y)=\mathcal{L}_{\textsf{MS}}(X\to Y)$ 
has properties $1), 2)$ and does not have the independence property $3)$.
\end{eg}

\begin{remark}
It is worth noting that Issa \textit{et al.} introduce \textit{maximal versions} of the logarithmic gain in \cite{8943950}. 
For example, they inrotoduced the \textit{variance leakage} $\mathcal{L}^{v}(X\to Y)$ as follows:
\begin{align}
\mathcal{L}^{v}(X\to Y) &:= \sup_{U-X-Y} \mathcal{L}_{\textsf{MS}}(U\to Y) \\ 
&= - \log (1-\rho_{m}(X; Y)), 
\end{align}
(see \cite[Def 10 and Lem 16]{8943950}) where 
\begin{align}
\rho_{m}(X; Y) &:= \sup_{\substack{f, g: \\ \vE[f(X)]=\vE[g(X)]=0, \\ \vE[f(X)^{2}]=\vE[g(X)^{2}]=1}} \vE[f(X)g(Y)]
\end{align}
is the \textit{maximal correlation}. Note that the variance leakage $\mathcal{L}^{v}(X\to Y)$ have all properties $1), 2)$ and $3)$ in Definition \ref{def:info_leakage} 
(see \cite[Prop 5.2]{asoodeh_phdthesis}). 
They also introduced a maximal version of all the logarithmic gain, called \textit{maximal cost leakage} $\mathcal{L}^{c}(X\to Y)$, as follows:
\begin{align}
\mathcal{L}^{c}(X\to Y) &:= \sup_{\substack{U-X-Y \\ \hat{U}, \ell\colon \mathcal{X}\times \hat{\mathcal{X}}\to \br_{+}}} \textsf{Lgain}^{\ell}(U; Y) \\ 
&= -\log \sum_{y} \min_{x\in \text{supp}(p_{X})} p_{Y\mid X}(y\mid x)
\end{align}
(see \cite[Def 11 and Thm 15]{8943950}). Note also that the maximal cost gain $\mathcal{L}^{c}(X\to Y)$ have have all properties $1), 2)$ and $3)$ in Definition \ref{def:info_leakage} 
(see \cite[Cor 5]{8943950}). 
In addition to these \textit{loss (cost) based information leakage} measures, they also introduced several \textit{utility}\footnote{Here we used the term `utility' in a statistical decision-theoretic sense. Note that Issa \textit{et al.} call `utility based information leakage' as `gain based information leakage'.} 
\textit{based information leakage} measures 
and showed relationships to the maximal information leakage $\mathcal{L}_{\text{MaxL}}(X\to Y)$. 
See \cite{8943950} for detail.
\end{remark}
For the logarithmic gain, a similar result as in Theorem \ref{thm:main_result_for_information_leakage} holds as follows.
\begin{cor}
For a classical loss $\ell(x,a)$, define a function as follows:
\begin{align}
{\textsf{LV}}_{\mathcal{L}}^{\ell}(R) 
&:= \log \inf_{a}\vE_{X}\left[\ell(X, a)\right] \notag \\ 
&- \inf_{\substack{p_{A\mid X}\colon \\ \mathcal{L}(X\to A)\leq R}} \log \vE_{X, A}\left[\ell(X, A)\right].
\end{align}
Then, following holds.
\begin{align}
\sup_{\mathcal{Y}} {\textsf{LV}}_{\mathcal{L}}^{\ell}(R; \mathcal{Y}) &= {\textsf{LV}}_{\mathcal{L}}^{\ell}(R).
\end{align}
\end{cor}

\section{Application to Privacy-Utility Trade-off}\label{sec:PUT}

In this section, we provide an interpretation of 
the achievability condition in Theorem \ref{thm:main_result_for_information_leakage} 
in the PUT context. 
We assume three parties: data curator (Alice), a legitimate user (Bob), and an adversary (Eve). 
Alice has the original data $X$ and disclose perturbed data $Y$ through a privacy mechanism $p_{Y\mid X}$ to prevent information leakage to Eve. A privacy constraint is represented 
as $\mathcal{L}(X\to Y)\leq R$, where the information leakage measure $\mathcal{L}(X\to Y)$ is chosen arbitrarily by Alice. 
While Bob's purpose of using the published data $Y$ is represented as an action, a deterministic decision rule and a loss function, i.e., 
$A=\delta(Y)$ and $\ell(x, a)$, respectively. 
Suppose that \textit{Alice knows the Bob's purpose of using the published data $Y$ before disclosure}.
We also assume that Bob make his action with the optimal decision rule 
$\delta^{\text{Bayes}}$ under the loss functions $\ell(x, a)$. 

In the situation above, Theorem \ref{thm:main_result_for_information_leakage} states that   
in order to maximize utility measured by $\textsf{gain}^{\ell}(X; Y)$ 
under the privacy constraint $\mathcal{L}(X\to Y)\leq R$,  
Alice should take the following steps: 
\begin{enumerate}
\item Find the channel $p^{*}_{A\mid X}$ such that 
\begin{align}
p^{*}_{A\mid X} = \arginf_{p_{A\mid X}\colon \mathcal{L}(X\to A)\leq R}\vE_{X, A}\left[\ell(X, A)\right]. 
\end{align}
\item Generate a random variable $\tilde{A}$ drawn to $p^{*}_{A\mid X}$.
\item Finally, disclose $Y=t(\tilde{A})$, a sufficient statistic of $\tilde{A}$ for $X$, to public. 
\end{enumerate}

{
\begin{remark}
When Alice assumes Eve's purpose of using $Y$, say $\delta^{*}_{\mathrm{eve}}$ and $\ell_{\mathrm{eve}}(x, y, a, \delta^{*}_{\mathrm{eve}})$, 
she can chose a privacy constraint as an average gain for Eve, i.e., 
$\mathcal{L}(X\to Y) := \textsf{gain}^{\ell_{\mathrm{eve}}}(X; Y)$. 
Note that she can even adopt the privacy constraint as the \textit{maximal gain} 
$\textsf{Mgain} ^{\ell_{\mathrm{eve}}}(X; Y)$ defined as follows, which is 
the inferential gain for using $Y$ in the most favorable situation for Eve.
\begin{definition}[Maximal gain]
For a standard loss $\ell(x, y, a, \delta^{*})$, the maximal gain of using $Y$ on $X$ for making an action $A$ is defined as
\begin{align}
&\textsf{Mgain}^{\ell}(X; Y) :=  
\vE_{Y}\left[\inf_{\delta^{*}(y, \cdot)}\vE_{X}\left[L(X, \delta^{*}(Y,\cdot))\right]\right] \notag \\ 
&\quad - \min_{y}\inf_{\delta^{*}(y, \cdot)} \vE_{X}\left[L(X, \delta^{*}(Y, \cdot))\relmiddle{|} Y=y\right]. \label{eq:max_gain}
\end{align}
\end{definition}
\end{remark}
Note that it follows immediately from \cite[Prop 23]{alvim_2018} that the maximal gain satisfies DPI.}

\section{Conclusion}\label{sec:conclusion}
In this study, we generalized the Stratonovich's VoI to 
formulate a problem of decision-making under a general information leakage constraint and a general loss function.
We derived upper bound for the VoI and showed weaker achievability condition than ever for a classical loss function. 
We presented an interpretation of these results in the PUT context and some extended results.
Future work includes deriving calculation algorithms for the upper bound.

\appendices

\section{Proof of Proposition \ref{prop:ub_for_information_leakage}} \label{proof:ub_for_information_leakage}
\begin{proof}
Define $\tilde{U}_{\mathcal{L}}^{\ell}(R; \mathcal{Y})$ and $\bar{U}_{\mathcal{L}}^{\ell}(R)$ as the second terms of the RHS 
in \eqref{eq:gen_VoI_for_classic_loss} and \eqref{eq:gen_VoI_fundamental_for_classic_loss}, respectively, i.e., 
\begin{align}
\tilde{U}_{\mathcal{L}}^{\ell}(R; \mathcal{Y}) &:= \inf_{\substack{p_{Y\mid X}\colon \\ {\mathcal{L}(X\to Y)}\leq R}}\vE_{Y}\left[\inf_{\delta^{*}(y, \cdot)} \vE_{X}\left[L(X, \delta^{*}(Y, \cdot))\relmiddle{|}Y\right]\right], \\ 
\bar{U}_{\mathcal{L}}^{\ell}(R; \mathcal{Y}) &:= \inf_{\substack{p_{Y\mid X}, \delta^{*}\colon \\ 
{\mathcal{L}(X\to A)} \leq R}} \vE_{X, Y}\left[L(X, \delta^{*}(Y, \cdot))\right].   
\end{align}
It suffices to show that $\tilde{U}_{\mathcal{L}}^{\ell}(R; \mathcal{Y}) \geq \bar{U}_{\mathcal{L}}^{\ell}(R; \mathcal{Y})$ for arbitrary alphabet $\mathcal{Y}$.
Define the privacy mechanism $\tilde{p}_{Y\mid X}$ and the optimal randomized decision rule $\tilde{\delta}^{*, \text{Bayes}}=\tilde{p}_{A\mid Y}$ as 
\begin{align}
\tilde{p}_{Y\mid X} &:= \arginf_{\substack{p_{Y\mid X}\colon \\ 
{\mathcal{L}(X\to Y)}\leq R}} \tilde{U}_{\mathcal{L}}^{\ell}(R; \mathcal{Y}), \label{eq:tilde_p_Y_given_X} \\
\tilde{\delta}^{*, \text{Bayes}}(y, a) &= \tilde{p}_{A\mid Y}(a\mid y) \\
&:= \arginf_{\delta^{*}(y, \cdot)} \sum_{x, a} \ell(x, y, a,\delta^{*}) \delta^{*}(y, a)\tilde{p}_{X\mid Y}(x\mid y),
\end{align}
where $\tilde{p}_{X\mid Y}(x\mid y) := \frac{p_{X}(x)\tilde{p}_{Y\mid X}(y\mid x)}{p_{Y}(y)}$. 
Since $X - Y - A$ forms a Markov chain for the distributions  $\tilde{p}_{Y\mid X}$ and $\tilde{\delta}^{*, \text{Bayes}}=\tilde{p}_{A\mid Y}$,  
\begin{align}
{\mathcal{L}(X\to A) \leq \mathcal{L}(X\to Y)} \leq R\label{eq:achievable}
\end{align}
holds from DPI \eqref{eq:DPI} and \eqref{eq:tilde_p_Y_given_X}.
Then from \eqref{eq:achievable}, 
\begin{align}
&\bar{U}_{\mathcal{L}}^{\ell}(R; \mathcal{Y}) = \inf_{\substack{p_{Y\mid X}, \delta^{*}\colon \\ 
{\mathcal{L}(X\to A)} \leq R}} \vE_{X, Y}\left[L(X, \delta^{*}(Y, \cdot))\right]  \\ 
&\leq \sum_{x, y, a}p_{X}(x)\tilde{p}_{Y\mid X}(y\mid x)\tilde{\delta}^{*, \text{Bayes}}(y, a) \ell(x, y, a, \tilde{\delta}^{*, \text{Bayes}}) \\
&= \tilde{U}_{\mathcal{L}}^{\ell}(R; \mathcal{Y}).
\end{align}
\end{proof}

\section{Proof of Theorem \ref{thm:main_result_for_information_leakage}} \label{proof:main_result_for_information_leakage}
Based on [2, Chapter. 9.7] and a refined proof in [19], we prove Theorem 1 as follows.
\begin{proof} 
Define ${U}_{\mathcal{L}}^{\ell}(R; \mathcal{Y})$ and ${U}_{\mathcal{L}}^{\ell}(R)$ as the second terms of RHS 
in \eqref{eq:gen_VoI_for_classic_loss} and \eqref{eq:gen_VoI_fundamental_for_classic_loss}, respectively, i.e., 
\begin{align}
{U}_{\mathcal{L}}^{\ell}(R; \mathcal{Y}) &:= \inf_{\substack{p_{Y\mid X}\colon \\ 
{\mathcal{L}(X\to Y)}\leq R}} \vE_{Y}\left[\inf_{a} \vE_{X}\left[\ell(X, a) \mid Y\right] \right], \\ 
{U}_{\mathcal{L}}^{\ell}(R) &:= \inf_{\substack{p_{A\mid X}\colon \\ 
{\mathcal{L}(X\to A)} \leq R}} \vE_{X,A}\left[\ell(X, A)\right].
\end{align}

\noindent
(Converse part): 
It suffices to show that $U_{\mathcal{L}}^{\ell}(R; \mathcal{Y})\geq U_{\mathcal{L}}^{\ell}(R)$ 
to prove $\textsf{V}_{\mathcal{L}}^{\ell}(R; \mathcal{Y})\leq \textsf{V}_{\mathcal{L}}^{\ell}(R)$ for arbitrary $\mathcal{Y}$.
This can be proved in a similar way to that in the proof of Proposition \ref{prop:ub_for_information_leakage} 
(see \cite[Appendix D]{9611484}). 

\noindent
(Achievable part): 
Let $\mathcal{Y} := t(\mathcal{A})$. 
It suffices to show that $U_{\mathcal{L}}^{\ell}(R; t(\mathcal{A}))\leq U_{\mathcal{L}}^{\ell}(R)$.
Define $p_{A\mid X}^{*}, p_{A}^{*}$ and $p_{X\mid A}^{*}$ as follows: 
\begin{align}
p^{*}_{A\mid X} &:= \argmin_{\substack{p_{A\mid X}\colon \\ \mathcal{L}(X\to A)\leq R}} 
\vE_{X, A}\left[\ell(X, A)\right], \label{eq:def_pax_origin} \\ 
p_{A}^{*}(a) &:= \sum_{x}p_{X}(x)p^{*}_{A\mid X}(a\mid x), \\ 
p^{*}_{X\mid A}(x\mid a)& := \frac{p_{X}(x)p^{*}_{A\mid X}(a\mid x)}{p_{A}^{*}(a)}. 
\end{align}
Let $\tilde{A}$ be a random variable drawn to $p_{A}^{*}$. 
Since $X - \tilde{A} - Y:=t(\tilde{A})$ forms a Markov chain, 
\begin{align}
\mathcal{L}(X\to Y) \leq \mathcal{L}(X\to \tilde{A})\leq R \label{eq:data_processing_ineq_direct_origin}
\end{align}
holds from DPI \eqref{eq:DPI} and \eqref{eq:def_pax_origin}. 
Now, define a privacy mechanism $p^{*}_{Y \mid X}$ as 
\begin{align}
p^{*}_{Y\mid X}(y\mid x) 
&:= \sum_{a} p^{*}_{A\mid X}(a\mid x)\one{y=t(a)}.
\end{align}
Then 
\begin{align}
U_{\mathcal{L}}^{\ell}(R; t(\mathcal{A})) 
&:= \inf_{\substack{p_{Y\mid X}\colon \\ 
\mathcal{L}(X\to Y) \leq R}} 
\vE_{Y}\left[\inf_{a} \vE_{X}\left[\ell(X, a)\mid Y\right]\right] \\ 
&\leq \vE_{Y}\left[\inf_{a} \vE_{X}^{p^{*}_{X \mid Y}}\left[\ell(X, a) \relmiddle{|} Y\right]\right],
\label{eq:ineq_mid_origin}
\end{align}
where the expectation $\vE_{X}^{p^{*}_{Y \mid X}}[\cdot]$ is taken over the distribution $p^{*}_{X \mid Y}(x | y) = p_{X}(x)p^{*}_{Y \mid X}(y | x)/p^{*}_{Y}(y)$.
Now, we will evaluate $\inf_{a} \vE_{X}^{p^{*}_{X\mid Y}}\left[\ell(X, a) \relmiddle{|} Y = t(a^{\prime})\right]$ from above.
\begin{align}
&\inf_{a} \vE_{X}^{p^{*}_{X\mid Y}}\left[\ell(X, a) \relmiddle{|} Y = t(a^{\prime})\right] \notag \\
&\overset{(*)}{=} \inf_{a} \vE_{X}^{p^{*}_{X\mid A}}\left[\ell(X, a) \relmiddle{|} \tilde{A}=a^{\prime}\right] \\ 
&\leq \vE_{X}^{p^{*}_{X\mid A}}\left[\ell(X, a^{\prime}) \relmiddle{|} \tilde{A} = a^{\prime}\right], 
\end{align}
where the equality $(*)$ follows from the sufficiency of $t(\tilde{A})$
\footnote{
It follows immediately from 
$p_{X\mid \tilde{A}}(x| a^{\prime}) = \sum_{y} p_{X, Y\mid \tilde{A}}(x, y | a^{\prime}) = 
\sum_{y}p_{Y\mid \tilde{A}}(y| a^{\prime}) p_{X\mid Y}(x| y) = p_{X\mid Y}(x|y)\one{y=t(a^{\prime})} = p_{X\mid Y}(x\mid t(a^{\prime}))$, where we used the sufficiency of $Y=t(\tilde{A})$ in 
the second equality.
}
. 
Thus we have 
\begin{align}
&\vE_{Y}\left[\inf_{a}\vE_{X}^{p^{*}_{X\mid Y}}\left[\ell(X, a)\right] \mid Y\right] 
= \vE_{\tilde{A}}\left[\inf_{a}\vE_{X}^{p^{*}_{X\mid A}}\left[\ell(X, a)\right] \mid \tilde{A}\right] \\
&\leq \vE^{p^{*}_{X\mid A}}_{X, \tilde{A}} \left[\ell(X, \tilde{A})\right]  
= \inf_{\substack{p_{A\mid X}\colon \\ 
\mathcal{L}(X\to A)\leq R}}\vE_{X, A}\left[\ell(X, A)\right]  \\ 
&= U_{\mathcal{L}}^{\ell}(R).
\end{align}
By combining with \eqref{eq:ineq_mid_origin},
$U_{\mathcal{L}}^{\ell}(R; t(\mathcal{A}))\leq U_{\mathcal{L}}^{\ell}(R)$.
  
\end{proof}

\section{Proof of Proposition \ref{prop:sufficient_statistic}} \label{proof:sufficient_statistic}
The sufficiency of $t(A) = A$ is trivial. 
To prove the sufficiency of $t(A) = (p_{X\mid A}(1 | A), p_{X\mid A}(2 | A), \dots, p_{X\mid A}(m | A))$, 
we first introduce the following lemmas.

\begin{lemma}[\text{\cite[Thm 6.12]{LC}}] 
Assume that a family of distributions $\{p_{A\mid X}(\cdot \mid x)\}_{x\in \mathcal{X}}$ have the same support.
Then 
\begin{align}
s(A) = \left( \frac{p_{A\mid X}(A\mid 2)}{p_{A\mid X}(A\mid 1)}, \dots, \frac{p_{A\mid X}(A\mid m)}{p_{A\mid X}(A\mid 1)} \right)
\end{align}
is a minimal sufficient statistic of $A$ for $X$. 
\end{lemma}

\begin{lemma} \label{lem:sufficient_function}
Let $T_{1}  = t_{1}(A)$ be a sufficient statistic of $A$ for $X$. 
If there exists a (measurable) function $f$ such that $T_{1} = f(t_{2}(A))$, then $T_{2} = t_{2}(A)$ is also sufficient for $X$.
\end{lemma}
\begin{proof}
The statement follows immediately from the Fisher's factorization theorem (see, e.g., \cite[Thm 6.5]{LC}) or DPI for Shannon's mutual information (see e.g., \cite[Eq (2.124)]{Cover:2006:EIT:1146355}). 
\end{proof}

\begin{lemma} 
\begin{align}
e(A) = \left( \frac{p_{X\mid A}(2\mid A)}{p_{X\mid A}(1\mid A)}, \dots, \frac{p_{X\mid A}(m\mid A)}{p_{X\mid A}(1\mid A)} \right)
\end{align}
is a (minimal) sufficient statistic of $A$ for $X$. 
\end{lemma}
\begin{proof}
\footnote{This proof is based on \cite[Prop 3.3]{Yatracos_2012}.}
Since $s(A) := \left( \frac{p_{A\mid X}(A\mid 2)}{p_{A\mid X}(A\mid 1)}, \dots, \frac{p_{A\mid X}(A\mid m)}{p_{A\mid X}(A\mid 1)} \right) = \left( \frac{p_{X}(1)}{p_{X}(2)}\cdot \frac{p_{X\mid A}(2\mid A)}{p_{X\mid A}(1\mid A)}, \dots, \frac{p_{X}(1)}{p_{X}(m)}\cdot \frac{p_{X\mid A}(m\mid A)}{p_{X\mid A}(1\mid A)} \right)$ is a function of $e(A)$, it follows from 
Lemma \ref{lem:sufficient_function} that $e(A)$ is also sufficient.
The minimality follows immediately as follows: For arbitrary $a, b\in \mathcal{A}$, it holds that $s(a) = s(b) \Longleftrightarrow e(a) = e(b)$. 
\end{proof}

Making use of these results, we prove Proposition \ref{prop:sufficient_statistic} as follows.
\begin{proof}
Since  $e(A) = \left( \frac{p_{X\mid A}(2\mid A)}{p_{X\mid A}(1\mid A)}, \dots, \frac{p_{X\mid A}(m\mid A)}{p_{X\mid A}(1\mid A)} \right)$ is a function of 
$t(A) = (p_{X\mid A}(1 | A), p_{X\mid A}(2 | A), \dots, p_{X\mid A}(m | A))$, 
from Lemma \ref{lem:sufficient_function}, 
$t(A)$ is also sufficient for $X$. 
\end{proof}

\section{Proof of Proposition \ref{prop:basic_property_gen_VoI}}\label{proof:basic_property_gen_VoI}

\begin{proof}
The property $1)$ is trivial. To prove the property $2)$, 
it suffices to show that $U(R) := \inf_{\substack{p_{A\mid X}\colon \\ 
\mathcal{L}{(X\to A) \leq R}}}\vE_{X, A}\left[\ell(X, A)\right]$ 
is convex (resp. quasi-convex) when $\mathcal{L}(X\to A) = \mathcal{L}(p_{X}, p_{A\mid X})$ is convex (resp. quasi-convex). 
We will only prove the convexity.
For arbitrary $0\leq \lambda \leq 1$ and $0\leq R_{1}, R_{2}\leq K(X)$, 
define 
\begin{align}
p_{A\mid X}^{*, 1} &:= \arginf_{\substack{p_{A\mid X}\colon \\ 
\mathcal{L}(X\to A)\leq R_{1}}}\vE_{X, A}\left[\ell(X, A)\right], \\ 
p_{A\mid X}^{*, 2} &:= \arginf_{\substack{p_{A\mid X}\colon \\ 
\mathcal{L}(X\to A)\leq R_{2}}}\vE_{X, A}\left[\ell(X, A)\right], \\
p_{A\mid X}^{*, \lambda} &:= \lambda p_{A\mid X}^{*, 1} + (1-\lambda) p_{A\mid X}^{*, 2}.
\end{align}
Then let denote 
$\mathcal{L}^{*, 1}(X\to A), \mathcal{L}^{*, 2}(X\to A)$ and 
$\mathcal{L}^{*, \lambda}(X\to A)$ as the $\alpha$-leakages 
defined by  
$p_{A\mid X}^{*, 1}, p_{A\mid X}^{*, 2}$ 
and $p_{A\mid X}^{*, \lambda}$, respectively. Then 
\begin{align}
\mathcal{L}^{*, \lambda}(X\to A) 
&\leq \lambda \mathcal{L}^{*, 1}(X\to A) \notag \\ 
&\qquad + (1-\lambda) \mathcal{L}^{*, 2}(X\to A)\\ 
&\leq \lambda R_{1} + (1-\lambda)R_{2}. \label{ineq:first}
\end{align}
Therefore, 
\begin{align}
U(\lambda R_{1} + (1-\lambda)R_{2})  
&\leq \vE^{p_{A\mid X}^{*, \lambda}}_{X, A}\left[\ell(X, A)\right] \\ 
&= \sum_{x, a}p_{X}(x)p_{A\mid X}^{*, \lambda}(a\mid x) \ell(x, a) \\ 
&= \lambda U(R_{1}) + (1-\lambda) U(R_{2}).  \label{ineq:second}
\end{align}
The quasi-convexity can be proved in a similar way.

To prove the property 3), it suffices to show that 
\begin{align}
\textsf{V}_{\mathcal{L}_{2}}^{\ell} (R; \mathcal{Y}) \leq \textsf{V}_{\mathcal{L}_{1}}^{\ell} (cR; \mathcal{Y})
\end{align}
for arbitrary alphabet $\mathcal{Y}$. 
To this end, define 
\begin{align}
p_{Y\mid X}^{*, 2} &:= 
\argsup_{\substack{p_{Y\mid X}\colon \\ 
{\mathcal{L}_{2}(X\to Y)}\leq R}} \textsf{gain}^{\ell}(X; Y)
\end{align}
for arbitrary alphabet $\mathcal{Y}$. 
Since 
\begin{align}
\mathcal{L}_{1}(p_{X}, p_{Y\mid X}^{*, 2}) \leq c \mathcal{L}_{2}(p_{X}, p_{Y\mid X}^{*, 2}) \leq cR, 
\end{align}
it holds that 
\begin{align}
\textsf{V}_{\mathcal{L}_{1}}^{\ell}(cR; \mathcal{Y}) &:= \sup_{\substack{p_{Y\mid X}\colon \\ \mathcal{L}_{1}(X\to Y)\leq cR}} \textsf{gain}^{\ell}(X; Y) \\ 
&\leq \textsf{gain}^{\ell}(p_{X}, p_{Y\mid X}^{*, 2}) = \textsf{V}_{\mathcal{L}_{2}}^{\ell} (R; \mathcal{Y}),
\end{align}
where $\textsf{gain}^{\ell}(p_{X}, p_{Y\mid X}^{*, 2}) := r(\delta^{*, \text{Bayes}}, p_{Y}^{*, 2}) - r(\delta^{*, \text{Bayes}}, p_{Y\mid X}^{*, 2})$ and $p_{Y}^{*, 2}(y) := \sum_{x}p_{X}(x)p_{Y\mid X}^{*, 2}(y\mid x)$. 
\end{proof}

\clearpage

\end{document}